\documentclass[12pt]{article}
\usepackage{amssymb,amscd}
\usepackage{verbatim}

\usepackage{xcolor}

\usepackage{amsmath,amssymb,graphicx,mathrsfs}   
\usepackage{enumerate}

\usepackage{tikz}
\usetikzlibrary{matrix}

\usepackage[all]{xy}

\usepackage{amssymb,amsfonts,amsthm,amsmath,calligra}
\usepackage{slashed}
\usepackage{yfonts}
\usepackage{mathrsfs,pifont}
\usepackage{float}

\usepackage[all]{xy}

\usepackage{tikz-cd}
\tikzcdset{every label/.append style = {font = \tiny}}

\usepackage{graphicx}
\usepackage{xcolor}

\usepackage{amssymb,amsfonts,amsthm,amsmath}
\usepackage[all]{xy}
\usepackage{slashed}
\usepackage{yfonts}
\usepackage{mathrsfs,pifont}

\usepackage{accents}

\usepackage[left=1in,right=1in]{geometry}

\let\frak\mathfrak

\def\>{\relax\ifmmode\mskip.666667\thinmuskip\relax\else\kern.111111em\fi}
\def\<{\relax\ifmmode\mskip-.333333\thinmuskip\relax\else\kern-.0555556em\fi}
\def\vsk#1>{\vskip#1\baselineskip}
\def\vv#1>{\vadjust{\vsk#1>}\ignorespaces}
\def\vvn#1>{\vadjust{\nobreak\vsk#1>\nobreak}\ignorespaces}

  \let\ssize\scriptstyle
\let\sssize\scriptscriptstyle

\let\Medskip\medskip
\def\medskip{\par\Medskip}
\let\Bigskip\bigskip
\def\bigskip{\par\Bigskip}

\let\Maketitle\maketitle
\def\maketitle{\Maketitle\thispagestyle{empty}\let\maketitle\empty}

\newtheorem{thm}{Theorem}[section]
\newtheorem{cor}[thm]{Corollary}

\newtheorem{prop}[thm]{Proposition}

\newtheorem{defn}[thm]{Definition}

\theoremstyle{definition}                                  
\numberwithin{equation}{section}

\theoremstyle{definition}

\let\mc\mathcal
\let\nc\newcommand

\let\la\lambda

\let\phi\varphi

\let\der\partial
\let\Hat\widehat

\let\Tilde\widetilde

\let\geq\geqslant

\let\leq\leqslant

\let\on\operatorname

\let\bs\boldsymbol

\def\C{{\mathbb C}}
\def\Z{{\mathbb Z}}

\def\F{{\mathbb F}}   

\def\+#1{^{\{#1\}}}

\def\beq{\begin{equation}}
\def\eeq{\end{equation}}
\def\be{\begin{equation*}}
\def\ee{\end{equation*}}

\nc{\bea}{\begin{eqnarray*}}
\nc{\eea}{\end{eqnarray*}}
\nc{\bean}{\begin{eqnarray}}
\nc{\eean}{\end{eqnarray}}

\nc{\Il}{{\mc I_{\bs\la}}}
\nc{\bla}{{\bs\la}}
\nc{\Fla}{\F_\bla}
\nc{\tfl}{{T^*\Fla}}
\nc{\GL}{{GL_n(\C)}}
\nc{\GLC}{{GL_n(\C)\times\C^*}}

\let\sd s 

\def\ddk_#1{\kk_{#1}\<\>\frac\der{\der\<\>\kk_{#1}}}

\def\bul{\mathbin{\raise.2ex\hbox{$\sssize\bullet$}}}
\def\intt{\mathchoice
{\mathop{\raise.2ex\rlap{$\,\,\ssize\backslash$}{\intop}}\nolimits}
{\mathop{\raise.3ex\rlap{$\,\sssize\backslash$}{\intop}}\nolimits}
{\mathop{\raise.1ex\rlap{$\sssize\>\backslash$}{\intop}}\nolimits}
{\mathop{\rlap{$\sssize\<\>\backslash$}{\intop}}\nolimits}}

\let\kk q 
\let\cc c

\let\Ko K

\def\GZ/{Gelfand-Zetlin}
\def\KZ/{{\slshape KZ\/}}
\def\qKZ/{{\slshape qKZ\/}}
\def\XXX/{{\slshape XXX\/}}

\nc{\A}{{\mc A}}

\nc{\hsl}{\widehat{{\frak{sl}_2}}}

\nc{\BC}{{ \mathbb C}}
\nc{\lra}{\longrightarrow}
\nc{\CO}{{\mathcal{O}}}
\nc{\BZ}{{ \mathbb Z}}
\nc{\hfn}{\hat{\frak{n}}}
\nc\Zs{{\Z/p^s\Z}}
\nc\Zo{{\Zs[z]^0}}
\nc\gr{{\on{gr}}}

\nc\fD{{\frak D}}

\newcommand{\matC}{\mathbb{C}}
\newcommand{\matN}{\mathbb{N}}

\newcommand{\matP}{\mathbb{P}}

\newlength{\dhatheight}

\usepackage{tikz}
\usepackage{ytableau}

\usetikzlibrary{decorations}
\usetikzlibrary{decorations.pathmorphing}
\usetikzlibrary{calc} 

\usepackage{hyperref}

\title{Capped Vertex Functions for $\text{Hilb}^n(\mathbb{C}^2)$}

\author{Jeffrey Ayers and Andrey Smirnov}

\begin{document}
\maketitle
\begin{abstract}
   We obtain explicit formulas for the $K$-theoretic capped descendent vertex functions of $\text{Hilb}^n(\mathbb{C}^2)$ for descendents given by the exterior  algebra of the tautological bundle. This formula provides a one-parametric deformation of  the generating function for normalized Macdonald polynomials. In particular, we show that the capped vertex functions are rational functions of the quantum parameter. 
\end{abstract}


\setcounter{footnote}{0}
\renewcommand{\thefootnote}{\arabic{footnote}}

\section{Introduction}
K-theoretic capped vertex functions were defined in \cite{Ok17} as partition functions of relative quasimaps to Nakajima varieties. 
In this paper we consider capped vertex functions for the variety $X=\mathrm{Hilb}^{n}(\C^2)$ given by Hilbert scheme of $n$-points in the plane $\C^2$. We consider vertex functions with a special type of descendents given by the exteriors powers of the tautological bundle. In this case we obtain an explicit combinatorial formula for the capped vertex function. 

It was conjectured and proved in several special cases \cite{PaPi10,PaPi12,Sm16} that the partition functions corresponding to the capped vertex functions with descendents are Taylor series expansions of  rational functions. The explicit formula we obtain in this paper confirms this conjecture for this special choice of descendents. In this section we overview the main definitions and  results.

\subsection{Content of the Paper} 

In \cite{Ok17}, Okounkov introduced the capped vertex function with descendents as the following partition function:

\begin{defn}(Section 7 of \cite{Ok17})
    The capped vertex with descendent $\tau$ for  a Nakajima quiver variety $X$ is the following generating function:
    \bean \label{capverdefin}
    \widehat{\textup{V}}^{(\tau)}(z):=\sum_{d} {\textsf{\textup{ev}}}_{p_2,*}\left(\textsf{\textup{QM}}^d_{rel\,p_2},\hat{\mathcal{O}}^d_{vir} \otimes \textsf{\textup{ev}}^{*}_{p_1}(\tau) \right)z^d \in K_{G}\left(X\right)_{}[[z]]
    \eean 
\end{defn}
Here $\textsf{\textup{QM}}^d_{rel\,p_2}$ denotes the moduli space of quasimaps from $\matP^1$ to a Nakajima variety $X$ relative to $p_2=\infty\in \matP^1$. By $\hat{\mathcal{O}}^d_{vir} \in K_{G}(\textsf{\textup{QM}}^d_{rel\,p_2})$ we denote the symmetrized virtual structure sheaf of the moduli space $\textsf{\textup{QM}}^d_{rel\,p_2}$. This moduli space is equipped with a proper map ${\textsf{\textup{ev}}}_{p_2}:\, \textsf{\textup{QM}}^d_{rel\,p_2} \to X$ and ${\textsf{\textup{ev}}}_{p_2,*}$ denotes the corresponding push-forward of the  $G$-equivariant K-theory groups, where $G$ is a certain symmetry group acting on both the moduli space and $X$. For a point $p_1=0 \in \matP^2$ we also have a map ${\textsf{\textup{ev}}}_{p_1}:\, \textsf{\textup{QM}}^d_{rel\,p_1} \to [X]$ where $[X]$ is the quotient stack corresponding to Nakajima variety $X$ (we recall that $X$ is a GIT quotient, while $[X]$ is a categorical quotient). Using the pull-back map $\textsf{\textup{ev}}^{*}_{p_1}$ for any class $\tau \in K_{G}([X])$ we can construct a K-theory class $\textsf{\textup{ev}}^{*}_{p_1}(\tau)$ on the moduli space.  The resulting vertex function is a generating function of the corresponding push-forwards over all degrees~$d$. The parameter $z$ counting the degrees of quasimaps is usually referred to as the {\it K\"ahler parameter}.

In this paper we study the capped vertex function for $X=\textup{Hilb}^n(\mathbb{C}^2)$ with descendents 
\bean \label{descdef}
\tau= \Lambda^{\bullet}_u (\mathcal{V}) = \sum_{k=1}^{n}  (-1)^k \Lambda^{k}(\mathcal{V}) u^k
\eean 
where $\mathcal{V}$ is the tautological bundle on the Hilbert scheme.

Let $T=(\matC^{\times})^2$ be a two-dimensional torus acting on $\matC^2$ by $(x,y)\to (t_1 x, t_2 y)$. We consider the induced action of $T$ on $\textup{Hilb}^n(\mathbb{C}^2)$ and the corresponding $T$-equivariant K-theory groups.
Recall that the set of $T$-fixed points is labeled by partitions of $n$: 
$$
\textup{Hilb}^n(\mathbb{C}^2)^{T}=\{\lambda:  |\lambda|=n \}.
$$
The classes of torus fixed points $[\lambda]$ form a basis of equivariant K-theory $K_{T}\left(\textup{Hilb}^n(\mathbb{C}^2)\right)$.

It is well-known that the equivariant cohomology and K-theory of the Hilbert schemes are equipped with an action of the Heisenberg algebra, see Section 8 of \cite{NakajimaLectures1} or \cite{FT2011} for a construction. 
Using this action $T$-equivariant  $K$-theories of $\textup{Hilb}^n(\mathbb{C}^2)$ can be identified with an infinite dimensional polynomial space,  called the Fock space. The classes of the $T$-fixed points are identified with the Macdonald polynomial normalized as in the work of Haiman~\cite{Ha00}:

\begin{thm}(\cite{FT2011}, \cite{NakajimaLectures1},\cite{Ha00})
There is an isomorphism of graded vector spaces 
\bean \label{fockdef}
\bigoplus_{n=0}^{\infty}K_{T}\left(\textup{Hilb}^n(\mathbb{C}^2)\right)_{loc}= \textsf{\textup{Fock}}:=\mathbb{Q}[p_1,p_2,...]\otimes_{\mathbb{Z}}\mathbb{Q}(t_1,t_2)
\eean
where the grading on the left side is by $n$ and on the right is by $\textrm{deg}(p_k)=k$. Under this isomorphism the K-theory classes of the torus fixed points are identified with the Macdonald polynomials (in Haiman's normalization):
$$
[\lambda] \longrightarrow H_{\lambda}
$$
\end{thm}

\noindent
The main result of this paper is the explicit formula of the following generating function 

\bean \label{partfundef}
F(z,y)=\sum_n \widehat{\textup{V}}_{\text{Hilb}^n(\mathbb{C}^2)}^{(\tau)}(z) y^n \in \textsf{\textup{Fock}}[[z,y]]
\eean
where $\widehat{\textup{V}}_{\text{Hilb}^n(\mathbb{C}^2)}^{(\tau)}(z)$ is the capped vertex function for the Hilbert scheme ${\text{Hilb}^n(\mathbb{C}^2)}$ with descendent (\ref{descdef}).  In the notations of the Fock space (\ref{fockdef}) we obtain:
\begin{thm}(Theorem \ref{mainformula})
The generating function (\ref{partfundef}) is a Taylor series expansion of the following function:
\bean \label{mainformulaintro}
F(z,y) =  \exp\left( \sum\limits_{k=1}^{\infty} \dfrac{y^k  }{k (1-t_1^{2k})(1-t_2^{2k})} ( ( 1 -  u^k ) \hbar^{2 k} p_{k}+z^k \hbar^{ 2 k} q^{-k} \dfrac{\hbar^k - \hbar^{-k}}{1-(z\hbar/q)^k} p_k  )\right)
\eean
where $\hbar=t_1t_2$, and $q$ denoted the equivariant parameters of torus $\mathbb{C}^{\times}$ acting on the source of the quasimaps. 
\end{thm}
When $z=0$ and $u=0$ this formula describes the generating function of the classes of structure sheaves ${\mathcal{O}}_{\textup{Hilb}^{n}(\matC^2)} \in K_{T}(\textup{Hilb}^{n}(\matC^2))$, which can be expressed as a sum of Macdonald polynomials $H_{\lambda}$ in the Fock space. Our formula, therefore, may be understood as a $(z,u)$-deformation of the well-known identity for the Macdonald polynomials:
\begin{prop}(Proposition \ref{macidentity})
   $$ \sum_{\lambda}\frac{H_{\lambda}}{\Lambda^{\bullet}(T_{\lambda}\textup{Hilb}^{|\lambda|}(\mathbb{C}^2))}y^{|\lambda|}=\exp\left(\sum\limits_{k=1}^{\infty} \dfrac{ y^k \hbar^{2 k} p_{k}}{k (1-t_1^{2k})(1-t_2^{2k})}  \right) $$
\end{prop}

Notice that the $y$-coefficients in the Taylor expansion of (\ref{mainformulaintro}) are a manifestly rational functions of $z$, and thus we get the following corollary

\begin{cor}(Corollary \ref{rationalcor})
The capped vertex function $\widehat{\textup{V}}_{\textup{Hilb}^n(\mathbb{C}^2)}^{(\tau)}(z)$ is a Taylor series expansion of a rational function in $z$:
    $$\widehat{\textup{V}}_{\textup{Hilb}^n(\mathbb{C}^2)}^{(\tau)}(z)\in \mathbb{Q}(t_1,t_2,z,q,u)$$
\end{cor}

We remark also that our main result (\ref{mainformulaintro}) can be written as
\bean \label{ook}
F(z,y) =  \mathsf{S}^{\bullet}\left(
\dfrac{y \hbar^2 p_1 }{(1-t_1^{2})(1-t_2^{2})} ( ( 1 - u )+z q^{-1} \dfrac{\hbar - \hbar^{-1}}{1-(z\hbar/q)} ) \right)
\eean
where $\mathsf{S}^{\bullet}$ denotes the 
symmetric algebra or the {\it plethystic exponential} of the argument, for instance see Section 2.1 of \cite{Ok17} for notations and definitions. We note also that structurally, our formula for descendents look similar to explicit formulas for two-legs capped vertex function obtained in \cite{KOO2011} which also have a convenient presentation as plethystic exponential of expression linear in $p_1$'s, see Theorem 1 in \cite{KOO2011}. This property of (\ref{ook}) implies that the collection $\widehat{\textup{V}}_{\textup{Hilb}^n(\mathbb{C}^2)}^{(\tau)}(z)$
labeled by $n \in \mathbb{N}$, with descendents (\ref{descdef}) is {\it factorizable} see Section 5.3.1 of \cite{Ok17} for definition. It would be interesting to find a direct geometric argument explaining this property. 

\subsection{Idea of the proof}
Let us outline the main ideas used in the proof of (\ref{mainformulaintro}).  It is well known that the Hilbert scheme 
$\text{Hilb}^n(\mathbb{C}^2)$ is a special case of a more general family of smooth symplectic varieties $\mathcal{M}(n,r)$ labeled by $n,r \in \matN$, known as instanton moduli spaces. When the rank of the instantons is $r=1$, we obtain the Hilbert scheme $\mathcal{M}(n,1)=\text{Hilb}^n(\mathbb{C}^2)$. The moduli space $\mathcal{M}(n,r)$ is an example of a Nakajima quiver variety and therefore we may study the capped vertex functions associated to $\mathcal{M}(n,r)$. 
Central to us is the result of Okounkov, see Theorem 7.5.23 in \cite{Ok17}, on large framing vanishing.  This theorem says that when $r$ is sufficiently large  these capped vertex functions are classical, i.e., all terms in the sum (\ref{capverdefin}) with $d>0$ vanish. For the descendents we study in this paper (\ref{descdef}) this happens when $r\geq 2$. This means that  for $\mathcal{M}(n,r)$ with $r\geq 2$ we know the corresponding capped vertex explicitly:
$$
\widehat{\textup{V}}^{(\tau)}_{\mathcal{M}(n,r)}(z)= \tau \in K_{G}(\mathcal{M}(n,r))
$$
and the right side is independent of $z$. 

The moduli space of rank $r=2$ instantons is equipped with an action of the torus $\matC^{\times}_a \cong \matC^{\times}$ such that the corresponding set of fixed points has the form:
\bean \label{limdecomp}
\mathcal{M}(n,2)^{\mathbb{C}_a^{\times}}= \coprod_{n_1+n_2=n} \textup{Hilb}^{n_1}(\mathbb{C}^2) \times \textup{Hilb}^{n_2}(\mathbb{C}^2).
\eean
Our idea is to extract the information about the capped vertex function for the Hilbert scheme from the limit of $a\to 0$ of the $\matC^{*}_a$ -equivariant capped vertex function of $\mathcal{M}(n,2)$, where $a$ denotes the equivariant parameter associated with the torus $\mathbb{C}_a^*$.

The main technical tool for computing this limit is the following result:
\begin{thm}(Section 7.4 of \cite{Ok17}) \label{CappingVertexIntro} The capping operator factors into a product
$$
\widehat{\textup{V}}^{(\tau)}(z) = \Psi(z) \textup{V}^{(\tau)}(z)
$$
where $\textup{V}^{(\tau)}(z)$ is the bare vertex function with descendent $\tau$ and $\Psi(z)$ is the capping operator. 
\end{thm}
The bare vertex function $\textup{V}^{(\tau)}(z)$ is defined similarly to (\ref{capverdefin}), with the moduli space $\textsf{\textup{QM}}^d_{rel\,p_2}$ replaced by the moduli space of quasimaps $\textsf{\textup{QM}}^d_{ns\,p_2}$ non-singular at $p_2$. The bare vertex functions are much simpler to compute using equivariant localization in K-theory. Using this explicit description of the bare vertex function we obtain the following result. 
\begin{prop}(Proposition \ref{vertexlimit})
We have the following limit of vertex functions:
\bean \label{limbare}
\lim\limits_{a\to 0} \, \CO(-1)\textup{V}^{\tau,(2)}(z)= \CO(-1)\textup{V}^{\tau,(1)}(z (-\hbar q)^{-1}) \otimes \CO(-1)\textup{V}^{\tau,(1)}(z (-\hbar q))
\eean
\end{prop}
Here and throughout the paper we use superscript $(r)$ to denote functions corresponding to $\mathcal{M}(n,r)$. For example, $V^{\tau,(2)}(z)$ in the left side of (\ref{limbare}) 
denotes the $\matC^{*}_a$-equivariant bare vertex function of for $\mathcal{M}(n,2)$.
The right side of (\ref{limbare}) corresponds to the product of the vertex functions for the Hilbert schemes $\mathcal{M}(n_i,1) = \textup{Hilb}^{n_i}(\C^2)$ given by the components of components the $\matC^*_a$-fixed set (\ref{limdecomp}). In words, this proposition says that in the limit $a\to 0$, the bare vertex function of $\mathcal{M}(n,2)$ factors, up to some shifts of the K\"ahler parameter $z$, to the product of the vertex functions for the Hilbert schemes. The appearance of the $\CO(-1)$, which denotes the operator of multiplication by the line bundle
$\CO(-1)$ in the equivariant K-theory, in the above formula is a result of a certain normalization we wish to have in the basis of fixed points. The vertex functions as defined in \cite{Ok17} normalized to that they have constant term of $\CO(1)$. For our purposes we wish for this power series to instead begin with the structure sheaf~$\CO_{\textup{Hilb}^{n}(\C^2)}$. The minor change in normalization is needed to ensure that the limits in the above proposition exist.\\

A similar factorization exists for the capping operator $\Psi(z)$. It is known that $\Psi(z)$ is a fundamental solution matrix of the quantum difference equation associated with the corresponding Nakajima variety \cite{OkSm}.  For the instanton moduli $\mathcal{M}(n,r)$ this equation is the quantum dynamical equation for the quantum toroidal algebra  $\mathcal{U}_{\hbar}({\ddot{\mathfrak{gl}_1}})$. Using this algebraic description, following \cite{Sm16} we obtain 
\begin{prop}(Proposition \ref{limitcapping})
The normalized capping operator $\Tilde{\Psi}^{(2)}(z)$ for the instanton moduli space $\mathcal{M}(n,2)$ has the following limit:
\begin{align}
 \lim\limits_{a\to 0} \,\Tilde{\Psi}^{(2)}(z)=Y(z)\,\Tilde{\Psi}^{(1)}(z\hbar^{-1})\otimes \Tilde{\Psi}^{(1)}(z\hbar)
\end{align}
where $\Tilde{\Psi}^{(1)}(z\hbar^{-1})\otimes \Tilde{\Psi}^{(1)}(z\hbar)$ in the right side denotes the tensor product of capping operators for the Hilbert schemes corresponding to the $\matC^{*}_a$ - fixed components (\ref{limdecomp}).
\end{prop}
The operator $Y(z)$ in this proposition is the fusion operator for the toroidal algebra $\mathcal{U}_{\hbar}({\ddot{\mathfrak{gl}_1}})$. In the proof of Proposition \ref{limitcapping} we give an explicit description for its action  on the Fock spaces:
\bean \label{yoperdef}
Y(z)=\exp{\left(\sum_{k=1}^{\infty}\frac{n_kK^{-k}\otimes K^{k}}{1-z^{-k}K^{-k}\otimes K^k}\alpha^0_{-k}\otimes \alpha^0_{k}\right)}
\eean
where $\alpha^0_k$ denote the generators of the horizontal Heisenberg subalgebra of $\mathcal{U}_{\hbar}({\ddot{\mathfrak{gl}_1}})$. These generators act on the Fock spaces  in a simple way (\ref{heisact}).\\

Using the factorization results summarized in the two previous  propositions,  in the limit $a\to 0$ of the Theorem \ref{CappingVertexIntro} we arrive at the following functional relation for the capped vertex functions of the Hilbert scheme:
\bean \label{funceq}
\tau \otimes 1=Y(z) \left(\widehat{\textup{V}}^{(\tau),(1)}(z) \otimes \dots \right)
\eean
where the left side does not have quantum corrections and therefore is known explicitly.  \\

Finally, since the operator $Y(z)$ given by (\ref{yoperdef}) is invertible and acts on the Fock spaces explicitly, we can solve the linear equation  (\ref{funceq}) for the first tensor component $\widehat{\textup{V}}^{(\tau),(1)}(z)$ which gives our main result.

\section{Acknowledgements} We thank H. Dinkins and A. Okounkov for their interest to this work and useful discussions. This work is supported by the NSF grants DMS - 2054527 and DMS-2401380.

\section{Background}
\subsection{Hilbert Schemes}

The Hilbert Scheme of $n$-points on $\mathbb{C}^2$ is the scheme parameterizing codimension $n$-ideals of $\mathbb{C}[x,y]$:
$$\text{Hilb}^n(\mathbb{C}^2) = \{\mathcal{J}\subset \mathbb{C}[x,y]: \dim (\mathbb{C}[x,y]/\mathcal{J})=n\}$$

Let $T\cong (\mathbb{C}^{\times})^2$ be the 2-dimensional algebraic torus acting on $\mathbb{C}^2$. This action scales coordinates:
 $$(x,y)\mapsto (t_1x,t_2y)$$

 And this extends to an action on $\text{Hilb}^n(\mathbb{C}^2) $, which scales the symplectic form $\omega\in H^2(\text{Hilb}^n(\mathbb{C}^2),\mathbb{C} )$ with character $\hbar=t_1t_2$. Let $A=\ker(\hbar)\subset T$ be the subtorus that fixes $\omega$.\\

 The fixed point set of $\text{Hilb}^n(\mathbb{C}^2)^T $ is a finite set, labeled by partitions of length $n$. For a partition $\lambda=(\lambda_1,\lambda_2,...,\lambda_k)$, with each $\lambda_i$ weakly decreasing, $\sum_{i=1}^k |\lambda|=n$, the corresponding fixed ideal is of the following form:
$$\mathcal{J}^T=\{y^{\lambda_1},xy^{\lambda_2},...,x^k\}$$
See Fig. 1 for an example:
 
\begin{figure}[h]
\centering

\begin{ytableau}
\none & y^2 & \none & \none \\
       \none & y & xy & \none \\
  \none & 1 & x & \none\\
  \none & \none & \none & \none
\end{ytableau}
\caption{An example of a torus fixed point for the Hilbert Scheme $\text{Hilb}^5(\mathbb{C}^2)$ corresponding to the partition $(3,2)$. The monomial ideal here is $\{y^3, xy^2, x^2\}$}
\end{figure}

 \subsection{Instantons}
 $\text{Hilb}^n(\mathbb{C}^2) $ appears as the rank 1 case of a more general moduli space: the moduli space of instantons of rank $r$. We describe this moduli space as the moduli of framed torsion free sheaves of rank $r$.\\

 Let $\mathcal{F}$ be a framed torsion free sheaf of rank $r$ over $\mathbb{P}^2$ with fixed second chern class $c_2(\mathcal{F})=n$. A framing of a sheaf is a choice of isomorphism 
 $$\phi: \left.\mathcal{F}\right|_{L_{\infty}}\rightarrow \mathcal{O}_{L_{\infty}}^{\oplus r}$$
 Where $L_{\infty}$ is the line at infinity in $\mathbb{C}^2$. The moduli space of these sheaves, denoted $\mathcal{M}(n,r)$ is known to arise as the Nakajima Quiver Variety of the framed double Jordan quiver. We quickly review the construction of this variety:\\

 Let $Q$ be the framed quiver with one vertex and one edge. Let $V=\mathbb{C}^n, W=\matC^r$. The representation space of $Q$ is
 $$Rep(Q)=\text{Hom}(W, V)\oplus \text{Hom}(V,V)$$
  This space has a natural $GL(V)$ action by conjugating the endomorphisms, and multiplying the vertical maps: $X\in \text{Hom}(V,V)$, $I\in \text{Hom}(W,V)$
$$X\mapsto gXg^{-1}, \ \ I\mapsto gI$$
This action induces a Hamiltonian action on 
$T^{*}Rep(Q)=Rep(Q)\oplus Rep(Q)^{*}$
with a moment map $$\mu: T^{*}Rep(Q)\rightarrow \mathfrak{g}^{*}$$
where $\mathfrak{g}=\text{Lie}(GL(V))$. With the choice of $GL(V)$ character $\theta: g \to \det(g)$ for the GIT quotient we have:
$$\mathcal{M}(n,r) = \mu^{-1}(0)//_{\theta} GL(V)$$
\begin{figure}
     \centering
\begin{tikzcd}
\mathbb{C}^n \arrow[loop, distance=3em, in=-10, out=50,looseness=1, "\textup{X}"] \arrow[loop, distance=3em, in=190, out= 120,looseness=10, "\textup{Y}"'] \arrow[d, bend left, "\textup{J}"] \\
\mathbb{C}^r \arrow[u, bend left, "\textup{I}"] 
\end{tikzcd}
\caption{The quiver giving rise to $\mathcal{M}(n,r)$}
 \end{figure}
In the special case of $r=1$ we recover the Hilbert scheme\cite{NakajimaLectures1,Na94}:
$\mathcal{M}(n,1)=\text{Hilb}^n(\mathbb{C}^2)$.
There is a natural torus action $A\cong (\mathbb{C}^{\times})^r$ on the framing $\mathbb{C}^r$. The $\text{i}^{th}$-component in the framing is scaled by $a_i$ under this action. Let $T=A\times (\mathbb{C}^{\times})^2$ with $\mathbb{C}^2$ scaling the plane coordinates with the same characters $t_i$ as above.

Let us fix a decomposition $W \cong \mathbb{C}^{r_1}\oplus \mathbb{C}^{r_2}$ where $r_1+r_2=r$ and 
let $\mathbb{C}_a^{\times}\subset A$ be a subtorus that acts on  $W \cong \mathbb{C}^{r_1}\oplus \mathbb{C}^{r_2}$
by scaling the second summand with a character $a$.  The fixed point set of this subtorus has the following form:
$$\mathcal{M}(n,r)^{\mathbb{C}^*}= \coprod_{n_1+n_2=n}\mathcal{M}(n_1,r_1) \times \mathcal{M}(n_2,r_2).$$
Acting by $T$ produces a fixed point set consisting of tuples of partitions of length $n$:

\begin{equation}
    \mathcal{M}(n,r)^T = \{(\lambda_1,...,\lambda_r): \sum |\lambda_i|=n \}
\end{equation}
The classes of fixed points $(\lambda_1,...,\lambda_n)$ form a basis in the localized $K$-theory of $\mathcal{M}(n,r)$. 


\subsection{Quasimaps and Vertex Functions}
%
Our main reference for this section is the notes \cite{Ok17}.

Let $\textsf{QM}^d$ denote the moduli space of degree $d$ quasimaps from $\mathbb{P}^1$ to $\mathcal{M}(n,r)$. Let $\mathbb{C}_q^{\times}$ be the standard torus acting on $\mathbb{P}^1$ by scaling homogeneous coordinates with weight $q$:

$$[x_1:x_2]\mapsto [qx_1:x_2]$$

There are two fixed points of this action: $p_1=0,p_2=\infty\in \mathbb{P}^1$, meaning that the only two nonsingular or regular points of the $\mathbb{C}_q^{\times}$-fixed quasimap may occur at these fixed points. Let $G=T\times \mathbb{C}_q^{\times}$ be the group acting on $\textsf{QM}^d$. There is a subset of the moduli space we will be primarily interested: the open subset of quasimaps of degree $d$ nonsingular at $p=0$ or $p=\infty$. Denote this open subset $\textsf{QM}^d_{ns \, p}\subset \textsf{QM}^d$. \\

The fixed point set $\left(\textsf{QM}^d_{ns\, p}\right)^G$ consists of pairs $(\lambda,d_{\lambda})$, where $\lambda=(\lambda_1,\lambda_2,...,\lambda_r)$ is a tuple of partitions such that $|\lambda|=n$ corresponding to a fixed point in $\mathcal{M}(n,r)$.  The second element in the tuple is the degree $d_{\lambda}$, which assigns a nonnegative integer $d_{\square}$ to each box  $\square\in \lambda$ in such a way that the data is organized into $r$-plane partitions \cite{Mac98}. These integers corresponding to the degree data must add up to $d$: $|d_{\lambda}|=\sum_{\square\in \lambda} d_{\square}=d$

There is an evaluation map from this open subset to $\mathcal{M}(n,r)$ mapping  a quasimap $f\in \textsf{QM}^d_{ns\, p}$  to  $f(p) \in \mathcal{M}(n,r)$. In general as $\textsf{QM}^d_{ns\, p}$ is not proper over $\mathcal{M}(n,r)$ there is no pushforward of this map to K-theory. We can avoid this issue by using localized K-theory instead. 

\begin{defn}(Section 7 of \cite{Ok17})
    The bare vertex of $\mathcal{M}(n,r)$ is the following generating function:
    
    $$\textup{V}(z):  = \sum_{d} \textsf{\textup{ev}}_{p_2,*}\left(\textsf{\textup{QM}}^d_{ns\, p_2},\hat{\mathcal{O}}^d_{vir} \right)z^d \in K_{G}\left(\mathcal{M}(n,r)\right)_{loc}[[z]] $$
where $\hat{\mathcal{O}}^d_{vir}$ is the  symmetrized virtual structure sheaf on $\textsf{\textup{QM}}^d_{ns\,p_1}$  
\end{defn}

For $p_1$, we do not have an evaluation map to to $\mathcal{M}(n,r)$ itself since we do not assume that $p_1$ is nonsingular. Rather we have an evaluation map to the quotient stack: $\textsf{ev}_{p_1}: \textsf{QM}^d_{ns\,p_2}\rightarrow \left[\mu^{-1}(0)/\prod GL(V_i)\right]$. For a class $\tau \in K_{G}(\left[\mu^{-1}(0)/\prod GL(V_i)\right])$  we have the following definition:

\begin{defn}(Section 7 of \cite{Ok17})\label{vertex}
    The bare vertex with descendent $\tau$ is the generating function:
\bean \label{desverte} \textup{V}^{(\tau)}(z):  = \sum_{d} \textsf{\textup{ev}}_{p_2,*}\left(\textsf{\textup{QM}}^d_{ns\, p_2},\hat{\mathcal{O}}^d_{vir} \otimes \textsf{\textup{ev}}_{p_1}^*(\tau)\right) z^d \in K_{G}\left(\mathcal{M}(n,r)\right)_{loc}[[z]]\eean 
\end{defn}
Interesting to us, the classes $\tau$ can be constructed as follows. We have the following inclusion of quotient stacks:
\bean \label{stackinlu}
\left[\mu^{-1}(0)/GL(V)\right] \subset  [T^{*}Rep(Q)/GL(V)]
\eean 
Since $T^{*}Rep(Q)$ is a vector space, which is contractible to a point we have
$$
K(\left[T^{*}Rep(Q)/GL(V)\right]) \cong K([pt/GL(V)])  \cong K_{GL(V)}(pt) = Rep(GL(V)) 
$$
where the last term is the ring of representations of $GL(V)$, i.e., is a ring of symmetric polynomials in $n$-variables. Thus, we can view such a symmetric polynomial as a K-theory class  on the quotient stack $\left[\mu^{-1}(0)/GL(V)\right]$ by pulling it back via inclusion (\ref{stackinlu}). The descendent (\ref{descdef}) which we analyze in this paper corresponds to
$$
\tau = \sum_{k=0}^{n} (-1)^k e_k u^k
$$
where $e_k$ is the $k$-th elementary symmetric polynomial.

Next, let $\textsf{QM}^d_{rel\, p_2}$ be the moduli space of degree $d$ quasimaps from $\mathbb{P}^1$ to $\mathcal{M}(n,r)$ {\it relative} to~$p_2$. In full analogy with the previous definition we have:

\begin{defn}(Section 7 of \cite{Ok17})
    The capped vertex with descendent $\tau$ is the generating function:
    \bean \label{cappedvert} \widehat{\textup{V}}^{(\tau)}(z):=\sum_{d} {\textsf{\textup{ev}}}_{p_2,*}\left(\textsf{\textup{QM}}^d_{rel\, p_2},\hat{\mathcal{O}}^d_{vir} \otimes \textsf{\textup{ev}}_{p_1}^*(\tau) \right)z^d \in K_{G}\left(\mathcal{M}(n,r)\right)_{}[[z]]\eean
\end{defn}
Note that, ${\textsf{\textup{ev}}}_{p_2}$ is proper for relative maps, so we land in non-localized equivariant $K$-theory \cite{Ok17}.\\

A remarkable rigidity theorem first proven by Okounkov in \cite{Ok17} is that for sufficiently large rank the capped vertex function obeys a property known as \textit{large frame vanishing}:

\begin{thm}(Theorem 7.5.23 in \cite{Ok17}) \label{LargeFrame}
    For every $\tau$ with $r$ sufficiently large, the quantum corrections to capped vertex function vanish
    $$\widehat{\textup{V}}^{(\tau)}(z)= \tau(\mathcal{V})\mathcal{K}^{1/2}$$
where $\mathcal{K}$ is the canonical bundle on $\mathcal{M}(n,r)$. 
\end{thm}

\subsection{Capping Operators}
The vertex functions (\ref{desverte}) can be computed very explicitly using equivariant localization in equivariant K-theory. The computation of the capped vertex function (\ref{cappedvert}) is much more delicate problem. The intermediate ingredient we need for this is {\it the capping operator}.

\begin{defn}(Section 8 of \cite{Ok17})\label{cap}
    The capping operator is the generating function:
    \bean  \label{cap}
    \Psi(z):= \sum_d \textsf{\textup{ev}}_{p_1,*}\otimes \textsf{\textup{ev}}_{p_2,*}\left(\textsf{\textup{QM}}^d_{\substack{rel\, p_1\\ ns \,p_2}},\hat{\mathcal{O}}^d_{vir}\right)z^d \in K_{G}\left(\mathcal{M}(n,r)\right)^{\otimes 2}_{loc}[[z]]\eean 
where ${\textup{QM}}^d_{\substack{rel\, p_1\\ ns \,p_2}}$ denotes the moduli space of degree $d$ quasimaps from $\mathbb{P}^1$ to $\mathcal{M}(n,r)$ with non-singular conditions at $p_2 \in \mathbb{P}^1$ and relative conditions at $p_1 \in \mathbb{P}^1$.
\end{defn}

The following theorem gives a relation between the three previous definitions:

\begin{thm}(Section 7.4 of \cite{Ok17}) The capping operator maps the vertex function to the capped vertex function:
$$
\widehat{\textup{V}}^{(\tau)}(z) = \Psi(z) \textup{V}^{(\tau)}(z)
$$    
\end{thm}
The operator $\Psi(z)$ can be computed explicitly as the fundamental solution of the \textit{quantum difference equation} \cite{OkSm}. This equation is described via objects coming from a certain quantum group acting on the equivariant $K$-theory of $\mathcal{M}(n,r)$.

\subsection{Elliptic Hall Algebra} This section will follow \cite{ScVa13}, see also section 7.1 of \cite{OkSm}.\\

Let $\textbf{Z}=\mathbb{Z}^2$, $\textbf{Z}^*=\textbf{Z}\setminus\{(0,0)\}$, $$\textbf{Z}^{+}=\{(i,j)\in \textbf{Z}, i>0, \text{ or } i=0, j>0\}$$ and $\textbf{Z}^{-}=-\textbf{Z}^{+}$
Let $\textbf{a}=(a_1,a_2)\in \textbf{Z}^{+}$, and define the degree to be the gcd of the integers $a_1,a_2$:
$\deg(\textbf{a})=\gcd(a_1,a_2)$.
Set $\epsilon_{\textbf{a}}=1$ if $\textbf{a}\in \textbf{Z}^{+}$, $\epsilon_{\textbf{a}}=-1$ if $\textbf{a}\in \textbf{Z}^{-}$. If $\textbf{a}, \textbf{b}$ are non-collinear elements, then set $\epsilon_{\textbf{a}, \textbf{b}}=\text{sgn}(\det(\textbf{a}, \textbf{b}))$.
Let $$n_k = \frac{(t_1^{k/2}-t_1^{-k/2})(t_2^{k/2}-t_2^{-k/2})(\hbar^{k/2}-\hbar^{-k/2})}{k}.$$

\begin{defn}
    The elliptic hall algebra $\mathcal{U}_{\hbar}({\ddot{\mathfrak{gl}_1}})$ is an associative unital algebra over $\mathbb{C}(t_1^{1/2},t_2^{1/2})$ generated by elements $e_{\textbf{a}}, K_{\textbf{a}}$ where $\textbf{a}\in {\bf Z}^{+}$ with the following relations:
    \begin{itemize}
        \item $K_{\textbf{a}}$ is central for all $\textbf{a}\in {\bf Z}^{+}$, and $K_{\textbf{a}+\textbf{b}}=K_{\textbf{a}}K_{\textbf{b}}$
        \item If $\textbf{a}, \textbf{b}$ are collinear, then $$[e_{\textbf{a}},e_{\textbf{b}}]=\delta_{\textbf{a}+ \textbf{b}}\frac{K_{\textbf{a}}^{-1}-K_{\textbf{a}}}{n_{\text{deg}(\textbf{a})}}$$
        \item If $\text{deg}(\textbf{a})=1$ and the triangle with vertices $\{(0,0),\textbf{a},\textbf{a}+\textbf{b}\}$ has no interior points then 
        $$[e_{\textbf{a}},e_{\textbf{b}}]=\epsilon_{\textbf{b},\textbf{a}}K_{\alpha(\textbf{a},\textbf{b})}\frac{\Psi_{\textbf{a},\textbf{b}}}{n_1}$$
        Where $$\alpha(\textbf{a},\textbf{b})=\begin{cases} 
      \dfrac{\epsilon_{\textbf{a}}(\epsilon_{\textbf{a}}(\textbf{a})+\epsilon_{\textbf{b}}(\textbf{b})-\epsilon_{\textbf{a}+\textbf{b}}(\textbf{a}+\textbf{b}))}{2} & \epsilon_{\textbf{a}, \textbf{b}}=1 \\
     & \\
      \dfrac{\epsilon_{\textbf{b}}(\epsilon_{\textbf{a}}(\textbf{a})+\epsilon_{\textbf{b}}(\textbf{b})-\epsilon_{\textbf{a}+\textbf{b}}(\textbf{a}+\textbf{b}))}{2} & \epsilon_{\textbf{a}, \textbf{b}}=-1
   \end{cases}$$
   And the elements $\Psi_{\textbf{a},\textbf{b}}$ are defined by 
   $$\sum_{k=0}^{\infty} \Psi_{k\textbf{a}}z^k = \exp\left(\sum_{l=1}^{\infty} n_l e_{l\textbf{a}}z^l \right).$$
    \end{itemize}
\end{defn}

\subsection{Heisenberg subalgebras and wall $R$-matrices}

For $w\in \mathbb{Q}\cup \{\infty\}$, let $d(w)$ and $n(w)$ be the denominator, and numerator of $w$, respectively, with $d(w)\geq 0$, and $d(w), n(w)$ coprime. If $w=\infty$, set $d(w)=0, n(w)=1$. 

The elements $$\alpha_k^w= e_{(d(w)k,n(w)k)}, \hspace{5mm} k\in \mathbb{Z}\setminus\{0\}$$
 generate a quantum Heisenberg subalgebra in $\mathcal{U}_{\hbar}({\ddot{\mathfrak{gl}_1}})$, which is called Heisenberg subalgebra with slope~$w$. \\

These subalgebras are equipped with the following upper triangular $R$-matrix: 
$$R_{w}^{+}=\exp\left(-\sum_{k=1}^{\infty} n_k \alpha_k^w \otimes \alpha_{-k}^w\right)$$
and lower triangular $R$-matrix 
$$R_{w}^{-}=\exp\left(-\sum_{k=1}^{\infty} n_k \alpha_{-k}^w \otimes \alpha_{k}^w\right)$$

\subsection{Hopf Structures}

The elliptic Hall algebra is a triangular hopf algebra, where different structures arise from a choice of slope. The coproduct was described in \cite{OkSm}, and has the following expression

$$\Delta_s:\mathcal{U}_{\hbar}({\ddot{\mathfrak{gl}_1}}) \rightarrow \mathcal{U}_{\hbar}({\ddot{\mathfrak{gl}_1}})\otimes \mathcal{U}_{\hbar}({\ddot{\mathfrak{gl}_1}}) $$
Explicitly this coproduct acts on the slope 0 Heisenburg subalgebra as follows, where we define $\Delta=\Delta_0$:
\begin{align*}
    &\Delta(\alpha^0_{-k}) = \alpha^0_{-k}\otimes 1 + K^{-k}\otimes \alpha^0_{-k}\\
    &\Delta(\alpha^0_{k}) = \alpha^0_{k}\otimes K^{k} + 1\otimes \alpha^0_{-k}\\
    &\Delta(K) = K\otimes K\\
\end{align*}
\subsection{Fock space Representations}
Recall that by Nakajima's geometric construction of Heisenberg algebra we have
 an isomorphism of graded vector spaces 
$$\bigoplus_{n=0}^{\infty}K_{T}\left(\text{Hilb}^n(\mathbb{C}^2\right)_{loc}= \textsf{\textup{Fock}}:=\mathbb{Q}[p_1,p_2,...]\otimes_{\mathbb{Z}}\mathbb{Q}(t_1,t_2)$$
where the grading of $\textsf{\textup{Fock}}$ is defined to be $\deg(p_k)=k$. As previously mentioned, the fixed point set of $\text{Hilb}^n(\mathbb{C}^2)$  is labeled by partitions of length $n$. The K-theory classes of torus fixed points are mapped by this isomorphism to the Macdonald polynomials $H_{\lambda}$.

The Fock space representation of $\mathcal{U}_{\hbar}({\ddot{\mathfrak{gl}_1}})$ with evaluation the parameter $a$ is defined by 
$$
ev_a: \mathcal{U}_{\hbar}({\ddot{\mathfrak{gl}_1}}) \longrightarrow \textrm{End}(\textsf{Fock})
$$
which is given explicitly in the basis of Macdonald polynomials $H_\lambda$ by
\bean \label{alphazer} \label{heisact} ev_a(\alpha^0_{k}): H_\lambda \to  \begin{cases}
    -k\dfrac{\partial}{\partial p_k}(H_\lambda) &k>0\\
    & \\ 
    \dfrac{-p_{-n} H_\lambda}{(t_1^{k/2}-t_1^{-k/2})(t_2^{k/2}-t_2^{-k/2})} & k<0
    \end{cases}\eean
by
$$
ev_{a}(\alpha^{\infty}_k) : H_{\lambda} \to  a^{-m} \dfrac{(-1)^k \textrm{sign}(k) }{1-t^{ k}_1} \Big(\sum_{i=0}^{\infty} t^{ k (\lambda_k-1)}_1 t_2^{ (k-1)} \Big) H_{\lambda}
$$
and by
$$
ev_a(K_{(1,0)}) : H_\lambda \to \hbar^{-1/2} H_\lambda, \ \ \ ev_a(K_{(0,1)}) : H_\lambda \to H_\lambda
$$

Since the horizontal $\alpha^{0}_k$ and vertical $\alpha^{\infty}_k$ Heisenberg subalgebras generate the whole algebra, these formulas define a representation of $\mathcal{U}_{\hbar}({\ddot{\mathfrak{gl}_1}})$ in the Fock space.  Further, we define the evaluation map $$ev^{(r)}:\mathcal{U}_{\hbar}({\ddot{\mathfrak{gl}_1}}) \rightarrow \text{End}({\textsf{Fock}}^{\otimes r})$$
by
\bean \label{evalrep}
ev^{(r)}(\alpha)=ev_{(a_1)}\otimes \cdots \otimes ev_{(a_r)}\left(\Delta^{(r)}(\alpha)\right).
\eean
i.e., by ${\textsf{Fock}}^{\otimes r}$ we always denote a tensor product of $r$ Fock modules with evaluation parameters $a_1,\dots, a_r$. 

Recall that by splitting the framing as $\mathcal{W}=a_1+\cdots + a_r$ via the action of torus $A$ we can obtain isomorphism of vector spaces:
\bean \label{tepro}
\bigoplus_{n=0}^{\infty} K_T{\left(\mathcal{M}(n,r)\right)_{loc}} \cong \textsf{\textup{Fock}}^{\otimes r} 
\eean 
The geometric action of $\mathcal{U}_{\hbar}({\ddot{\mathfrak{gl}_1}})$ on (\ref{tepro}) constructed in \cite{OkSm} coincides with the one defined by evaluation map (\ref{evalrep}).

\subsection{Wall crossing operators, and the quantum difference operator}

Consider the elements defined in section 7 of \cite{OkSm}

\bean  \label{wallcross}
B_w(z) = :\exp\left(\sum_{k=0}^{\infty} \frac{n_k \hbar^{-krd(w)/2}}{1-z^{-kd(w)}q^{kn(w)}\hbar^{-krd(w)/2}}\alpha_{-k}^w\alpha_k^w\right):  
\eean
called the wall crossing operator for wall $w\in \mathbb{Q}\cup \{\infty\}$. Here the symbol :: means to first act with all operators $\alpha_k^w$ for $k>0$, i.e., we assume normal ordering.

The quantum difference operator is defined as the operator given by the following product
$$\textsf{\textup{M}}(z) =\CO(1)\prod_{-1\leq w< 0}^{\leftarrow} B_{w}(z)$$

Here, the left facing arrow denotes the order of the product.  Using the Fock space representation of $\mathcal{U}_{\hbar}({\ddot{\mathfrak{gl}_1}})$ we can evaluate the operators $\textsf{\textup{M}}(z)$ as certain  operators acting in the Fock space. Moreover, it is clear from (\ref{wallcross}) that the action of $\textsf{\textup{M}}(z)$ preserves the degree in the Fock space, thus  the degree $n$ block of $\textsf{\textup{M}}(z)$ acts as a certain operator in $K_{T}(\mathcal{M}(n,r))$. The capping operator (\ref{cap}) for $\mathcal{M}(n,r)$ can be computed as the fundamental solution matrix of the quantum difference equation

\begin{thm}(Section 7 of \cite{OkSm}) 
 The capping operator (\ref{cap}) for $K_{T}(\mathcal{M}(n,r))$ is the unique fundamental solution matrix of the quantum $q$-difference equation
    $$\Psi^{(r)}(zq)\mathcal{O}(1)=\textsf{\textup{M}}(z)\Psi^{(r)}(z)$$
 normalized by  $\Psi^{(r)}(0)=Id$, where $Id$ is the identity matrix.    
\end{thm}

\subsection{K-theoretic Stable Envelopes \label{stabsec}}
We briefly recall the basic definition and properties of K-theoretic Stable Envelopes. See Section 2.1 of \cite{OkSm}, or Sections 9.1-9.2 of \cite{Ok17} for a more detailed exposition.\\

As before let $T= A\times (\mathbb{C}^{\times})^2$ be the torus acting on $\mathcal{M}(n,r)$. Let $\hbar$ be the character of symplectic form and $A=\ker(\hbar) \subset T$ be the subtorus preserving the symplectic form. The K-theoretic stable envelope is a map between equivariant K-theories
$$\text{Stab}_{(r)}: K_{T}(\mathcal{M}(n,r)^A)\rightarrow K_{T}(\mathcal{M}(n,r))$$
which is uniquely determined by the following three choices.

First,  let  $\sigma \in \textrm{cochar}(A)$ be a generic cocharacter. The choice of $\sigma$  provides a decomposition 
$$
T_{\lambda}  \mathcal{M}(n,r) = N_{\lambda}^{+}\oplus N_{\lambda}^{-}
$$
of the tangent space at a torus fixed point $\lambda$ into attracting and repealing directions. Here $N_{\lambda}^{+}$ and $N_{\lambda}^{-}$ denote the subspaces with $A$-characters which are positive and, respectively, negative on $\sigma$. This choice also fixes the attracting sets to fixed points:
$$\text{Attr}(\lambda) =\{x\in \mathcal{M}(n,r): \lim\limits_{z\rightarrow 0} \sigma(z)\cdot x = \lambda\}$$

The full attracting set $\text{Attr}^f(\lambda)$ of a torus fixed point $\lambda$ is, by definition, the smallest closed subset of $X$ containing $\text{Attr}(\lambda)$ that is closed under $\text{Attr}(\cdot)$.

The second choice is a polarization, which is a class $T^{1/2}\in K_{T}(\mathcal{M}(n,r))$ satisfying 
$$
T\mathcal{M}(n,r)= T^{1/2}+\hbar\left(T^{1/2}\right)^{*}, 
$$
i.e., $T^{1/2}$ is an equivariant ``half'' of the tangent space.

Let $\mathcal{W}$ be the tautological bundle associated to the framing of $\mathcal{M}(n,r)$. This is a trivial rank $r$ vector bundle with the following restrictions to fixed points:
$$\mathcal{W}_{\lambda}=a_1+\cdots+ a_r \in K_{T}(pt)$$

Let $\mathcal{V}$ be the rank $n$ tautological bundle of $\mathcal{M}(n,r)$. The weights of $\mathcal{V}$ at a fixed point $\lambda=(\lambda_1,\lambda_2,...,\lambda_r)$ are given by 
$$\mathcal{V}_{\lambda}= \sum_{\square \in (\lambda_1,\lambda_2,...,\lambda_r)}\varphi_{\lambda_1,\lambda_2,...,\lambda_r}(\square) \in K_{T}(pt)$$
where \begin{equation}\label{weights}
\varphi_{\lambda_1,\lambda_2,...,\lambda_r}(\square)= a_{n(\square)}t_1^{y(\square)}t_2^{x(\square)}
\end{equation}
with $n(\square)=i$ if $\square\in \lambda_i$, and $x(\square), y(\square)$ denote the standard coordinates of a box $\square$ in the Young diagram $\lambda_i$. The function $\varphi_{\lambda_1,\lambda_2,...,\lambda_r}(\square)$ gives the weights associated to a box in a tuple of partitions.
For $\mathcal{M}(n,r)$ a polarization can be taken to be
\begin{equation}\label{tangent}
T^{1/2}=\mathcal{W}^{*}\otimes \mathcal{V}+\frac{1}{t_2}\mathcal{V}^{*}\otimes \mathcal{V}-\frac{1}{\hbar}\mathcal{V}^{*}\otimes \mathcal{V}-\mathcal{V}^{*}\otimes \mathcal{V}\end{equation}
    The symbol $*$ means taking dual in equivariant $K$-theory. This choice of polarization is canonical and exists for any Nakajima variety, see Section 2.2.7 of \cite{MO18}.

Finally, we require a choice of a ``slope'' which is a fractional line bundle $s\in \text{Pic}(\mathcal{M}(n,r))\otimes_{\mathbb{Z}}\mathbb{Q} \cong \mathbb{Q}$ which should be suitably generic, see Section 2.1 of \cite{OkSm} for details.

\begin{prop}(\cite{Ok17},\cite{OkSm}) Let $T, A$ be as above, then for an arbitrary choice of a character $\sigma$, polarization $T^{1/2}$, and slope $s$, there exists a unique map $$\text{Stab}_{(r)}: K_{T}(\mathcal{M}(n,r)^A)\rightarrow K_{T}(\mathcal{M}(n,r))$$ satisfying the following three axioms:

\begin{itemize}
    \item Support condition: $\textup{supp}(\textup{Stab}_{(r)})\subset \textup{Attr}^f$
    \item Degree condition: 
  
    $$\deg_A \left.\textup{Stab}_{(r)}\right|_{F_2 \times F_1} \otimes \left.s\right|_{F_1} \subset \deg_A \left.\textup{Stab}_{(r)}\right|_{F_2\times F_2} \otimes \left.s\right|_{F_2}$$
    Where $F_1, F_2$ are fixed components, and $\deg_A$ is the degree of a Laurent polynomial given by the Newton polygon:
    $$\deg_A \sum f_{\mu} z^{\mu} = \textup{Convex Hull}\left(\{\mu, f_{\mu}\neq 0\}\right)$$
    \item Normalization condition: $$\left.\textup{Stab}_{(r)}\right|_{F\times F} =(-1)^{rk T^{1/2}_{>0}}\left(\frac{\det N_{\lambda}^{-}}{\det T^{1/2}_{\neq 0}}\right)^{1/2}\otimes\Lambda^{\bullet}(N_{\lambda}^{-}) $$
\end{itemize}
    
\end{prop}

\section{Results for $r=2$ splitting}

Using localization in rank $r$ the relation between the capped and bare descendent vertex functions in the basis of torus fixed points for $\mathcal{M}(n,r)$ is
\begin{equation} \label{cappdef}
     \widehat{\textup{V}}^{(\tau)}(z) ={\Lambda^{\bullet}(T\mathcal{M}(n,r)^{\vee})}\Psi(z/(-q)^r) {\Lambda^{\bullet}(T\mathcal{M}(n,r)^{\vee})}^{-1}\textup{V}^{(\tau)}(z)
\end{equation}
 
Where ${\Lambda^{\bullet}(T\mathcal{M}(n,r)^{\vee})}$ is the normalized matrix of tangent weights for rank $r$.
In particular, for $X=\textup{Hilb}^n(\C^2)$ we have 
\bean \label{cappedrank1}
\widehat{\textup{V}}^{(\tau)}(z) ={\Lambda^{\bullet}(TX^{\vee})}\Psi(-z/q) {\Lambda^{\bullet}(TX^{\vee})}^{-1}\textup{V}^{(\tau)}(z)
\eean
The difference equation satisfied by $\Psi(z)$ is
\bean \label{qdeq}
\Psi(z/q^2) \CO(-1) = \textsf{M}(z) \Psi(z)
\eean

\subsection{Limit of multiplication by $\CO(-1)$, and of the matrix of tangent weights}

Let us fix a decomposition $r=r_1+r_2+\dots+r_m$, and the corresponding split of the framing space $W=\mathbb{C}^r =
\mathbb{C}^{r_1} \oplus \mathbb{C}^{r_2}\oplus \dots \oplus \mathbb{C}^{r_m}$. To such a split, we associate a cocharacter $\sigma: \mathbb{C}^{\times} \to A$, such that $\mathbb{C}^{\times}$
acts on the summands of $W$ with different weights, for instance, by scaling $\mathbb{C}^{r_i}$ with $a^i$, where $a$ is the coordinate on $\mathbb{C}^{\times}$.  
We note that for this choice 
\bean \label{torusfpc}
\mathcal{M}(n,r)^{\mathbb{C}^{*}}  = \coprod\limits_{n_1+\dots+n_m=n}\, \mathcal{M}(n_1,r_1) \times \dots \times \mathcal{M}(n_m,r_m)
\eean 
As we discussed in Section \ref{stabsec}, this choice provides a stable envelope map, which we denote by $\textup{Stab}_{(r_1,r_2,\dots,r_m)}$:
\bean \label{stabs}
\textup{Stab}_{(r_1,r_2,\dots,r_m)}: K_{T}\Big(\coprod\limits_{n_1+\dots+n_m=n}\, \mathcal{M}(n_1,r_1) \times \dots \times \mathcal{M}(n_m,r_m)\Big)\rightarrow K_T(\mathcal{M}(n,r))
\eean 
For example, the maximal split $r=1+\dots+1$ corresponds to a map
 $$\textup{Stab}_{(1,\dots, 1)}:K_T\Big(\coprod\limits_{n_1+\dots+n_r=n}\text{Hilb}^{n_1}(\mathbb{C}^2)\times \cdots \times  \text{Hilb}^{n_r}(\mathbb{C}^2)\Big)\rightarrow K_T(\mathcal{M}(n,r))$$
where $\text{Hilb}^{n_i}(\mathbb{C}^2) =\mathcal{M}(n_i,1)$.  
We also note that in this notation 
$$
\textup{Stab}_{(1)} : K_{T}(\text{Hilb}^n(\mathbb{C}^2)) \rightarrow  K_{T}(\text{Hilb}^n(\mathbb{C}^2)) 
$$
is trivial, i.e., the identity map. All such splittings are in agreement with each other via the following
\textit{triangle lemma}:
\begin{prop} (Prop 9.2.8 of \cite{Ok17}) The following diagram commutes:
\begin{footnotesize}
    \[\begin{tikzcd}[cramped]
	{K_T(\coprod\limits_{n_1+n_2=n} \mathcal{M}(n_1,r_1)\times \mathcal{M}(n_2,r_2))} &&& {K_T(\mathcal{M}(n,r)) } \\
	\\
	& {K_T\Big(\coprod\limits_{n_1+\dots+n_r=n}\textup{Hilb}^{n_1}(\mathbb{C}^2)\times \cdots \times  \textup{Hilb}^{n_r}(\mathbb{C}^2)\Big)}
	\arrow["{\textstyle \textup{Stab}_{(r_1,r_2)}}", from=1-1, to=1-4]
	\arrow["{\textstyle \textup{Stab}_{\underbrace{\small{(1,...,1)}}_{r}}}"', from=3-2, to=1-4]
	\arrow["{\textstyle \textup{Stab}_{\underbrace{(1,...,1)}_{r_1}}\otimes \textstyle \textup{Stab}_{\underbrace{(1,...,1)}_{r_2}}}", from=3-2, to=1-1]
\end{tikzcd}\]
\end{footnotesize}
\end{prop}

Define $\Delta_{(r_1,\dots,r_m)}$ to be the diagonal of the stable envelope map (\ref{stabs}). By the normalization axiom, it is equal to
\begin{equation}\label{delta}
   \Delta_{(r_1,\dots,r_m)}=\hbar^{-n}\left(\frac{\det N^{-}}{\det T_{\neq 0}^{1/2}}\right)^{1/2}\otimes \Lambda^{\bullet}(N^{-})
\end{equation}
where $N^{-}$ denotes the repelling part of the normal bundle to the torus fixed point components (\ref{torusfpc}) inside $\mathcal{M}(n,r)$. \\

We are specifically interested in the rank $r=2$ case, and the splitting $r=1+1$. We assume that the torus $\mathbb{C}^{*}$ scales the first summand of the framing with character $a$ and acts trivially in the second summand.

By $\CO(-1)$ we denote the operator of multiplication by the corresponding line bundle in the equivariant K-theory. Specifically, in the basis of torus fixed points, this operator has the following eigenvalue:
\begin{equation} \label{part1}
\left.\CO(-1)\right|_{(\lambda_1,\lambda_2)}= a^{-|\lambda_1|}\prod_{k=1}^2 \prod_{(i,j)\in \lambda_k} t_1^{1-j}t_2^{1-i}
\end{equation} 
where $\lambda = (\lambda_1, \lambda_2)$ denote a torus fixed point.  By the Kunneth formula we have an isomorphism of equivariant $K$-theories
$$K_T(\mathcal{M}(n_1,r_1)\times \mathcal{M}(n_2,r_2))\cong K_T(\mathcal{M}(n_1,r_1)) \otimes K_T(\mathcal{M}(n_2,r_2)).$$
Accordingly, we will use the notations for diagonal matrices with the following eigenvalues:
\bean \label{parts2}
\left.\left(\CO(-1)\otimes 1\right)\right|_{\lambda}=\prod_{(i,j)\in \lambda_1} t_1^{1-j}t_2^{1-i}, \ \ \ \left.\left(1\otimes \CO(-1) \right)\right|_{\lambda}=\prod_{(i,j)\in \lambda_2} t_1^{1-j}t_2^{1-i} 
\eean

\begin{prop} \label{prop1}
We have the following limit
$$
\lim\limits_{a\to 0} \, \Delta_{(1,1)}^{-1}\mathcal{O}(-1)a^{2n}= \hbar^n \left\{{\hbar}^{n_2}\CO(-1)\otimes 1\right\}
$$

\end{prop}
\begin{proof}
    By (\ref{part1}) and (\ref{parts2}) we have: $\CO(-1)= {a^{-n_1}} (\CO(-1)\otimes \CO(-1))$ where $|\lambda_1|=n_1$.  We also compute:

    \begin{align*}
T^{1/2}&=\hbar\mathcal{W}^*\otimes \mathcal{V}= \hbar\left(V_1 a + V_2\right)\left(\frac{1}{a}+1\right)\\ &=\hbar\left(\sum_{(i,j)\in\lambda_1}\left(t_1^{j-1}t_2^{i-1}+at_1^{j-1}t_2^{i-1}\right)+\sum_{(i,j)\in\lambda_2}\left(t_1^{j-1}t_2^{i-1}+a^{-1}t_1^{j-1}t_2^{i-1}\right)\right)
    \end{align*}
The repelling part corresponds to the terms with negative power of $a$:
    $$T^{1/2}_{<0}= \hbar\sum_{(i,j)\in \lambda_2}a^{-1}t_1^{j-1}t_2^{i-1}$$
Thus, taking determinant we get: $$\det(T^{1/2}_{<0}) = \frac{\hbar^{n_2}}{a^{n_2}} \Big(1\otimes \CO(1) \Big)$$    
where $n_2=|\lambda_2|$. Thus, from (\ref{delta}) we compute:
$$\Delta_{(1,1)}^{-1}=\hbar^{n}\frac{\det(T^{1/2}_{<0})}{\Lambda^{\bullet}(N_{\lambda}^{-})}=\hbar^{-n}\frac{1}{\Lambda^{\bullet}(N_{\lambda}^{-})}{\left(\frac{\hbar^{n_2}}{a^{n_2}}1\otimes \CO(1)\right)}$$
Combining these, recalling that $n=n_1+n_2$, and simplifying we see
\begin{align*}
\Delta_{(1,1)}^{-1}\mathcal{O}(-1)a^{n}&=\hbar^{n}\frac{1}{\Lambda^{\bullet}(N_{\lambda}^{-})}{\frac{\hbar^{n_2}}{a^{n_2}} \left(1\otimes \CO(1)\right)}\, \dfrac{1}{a^{n_1}}\, \left(\CO(-1)\otimes \CO(-1) \right)a^{n}\\
      &=\hbar^{n}\frac{1}{\Lambda^{\bullet}(N_{\lambda}^{-})}{\left({\hbar^{n_2}}\CO(-1)\otimes 1\right)}
   \end{align*}
We also have $\Lambda^{\bullet}(N_{\lambda}^{-})=\prod_{w\in N_{\lambda}^{-}}(1-w^{-1})$ with repelling weights $w\to\infty$ as $a\to 0$.
Thus $\lim\limits_{a\rightarrow 0}\dfrac{1}{\Lambda^{\bullet}(N_{\lambda}^{-})}=1$. Combining all this together in the limit $a\rightarrow 0$ we arrive at the desired expression 
$$\lim\limits_{a\rightarrow 0}\Delta_{(1,1)}^{-1}\mathcal{O}(-1)a^{n}=\hbar^{n}{\left({\hbar^{n_2}}\CO(-1)\otimes 1\right)}$$

\end{proof}

\begin{prop} We have the following limit
   $$ \lim\limits_{a\to 0 }\, \Delta_{(1,1)}^{-1}\cdot \textup{Stab}_{(1,1)} = R_{w_0}^{+}$$
\end{prop}
where $R_{w_0}^{+}$ denotes the zero-wall R-matrix. 
\begin{proof}
    By section 2.3.3 of \cite{OkSm} we can write the Stable Envelope as a product of $R$-matrices, as we fixed the positive chamber we have 
    $$\textup{Stab}_{(1,1)}=\textup{Stab}_{+,\infty}\cdots R_{w_2}^{+}R_{w_1}^{+}R_{w_0}^{+}$$
    By definition $\textup{Stab}_{+,\infty}=\Delta_{(1,1)}$ (section 2.3.3 of \cite{OkSm}). Thus, multiplying by the inverse of $\Delta_{(1,1)}$ leaves us with a product of wall $R$-matrices.     By equation (31) of \cite{OkSm} the wall $R$-matrices have 1 on the diagonal, and all of the off diagonal terms have positive powers of $a$ except for the zeroth wall. In the limit of $a\rightarrow 0$ all nonzero wall $R$-matrices go to 1, except the zeroth wall $R_{w_0}^{+}$.
\end{proof}

Let $J(z)$ denote the fusion operator \cite{OkSm}, which is expressed as the following sum:

\bean \label{fusi}
J(z)=\exp\left(-\sum_{k=1}^{\infty}\frac{n_kK^{-k}\otimes K^{k}}{1-z^{k}K^{-k}\otimes K^k}\alpha^{0}_{-k}\otimes \alpha^{0}_{k}\right)
\eean

\begin{prop} The wall 0 $R$-matrix is equal to the fusion operator at 0.
    $$R_{w_0}^{+}= J(0)$$
\end{prop}
\begin{proof}
    See Proposition 8 of \cite{OkSm}.
\end{proof}

\begin{cor}
We have the following limit of the stable envelope:
$$
\lim\limits_{a\to 0 }\, \Delta_{(1,1)}^{-1}\cdot \textup{Stab}_{(1,1)} =  J(0) 
$$
\end{cor}

\begin{proof}
    Apply the previous two propositions. 
\end{proof}



\subsection{Limit of Capping Operator}
 Let $\Psi^{(r)}(z)$ be the rank $r$ capping operator in what follows. Let $N_{r}$ denote the tangent bundle over $\mathcal{M}(n,r)$. 

\begin{prop}\label{limitcapping}
We have the following limit of the rank 2 fundamental solution matrix:
\begin{align} \nonumber
    &J(z)J(0)^{-1} \times \lim\limits_{a\to 0} \, \Delta_{(1,1)}^{-1} \Lambda^{\bullet}(N_2) \Psi^{(2)}(z) \Lambda^{\bullet}(N_2)^{-1} \Delta_{(1,1)}\\
    &=\Lambda^{\bullet}(N_1)\Psi^{(1)}(z\hbar)\Lambda^{\bullet}(N_1)^{-1}\otimes \Lambda^{\bullet}(N_1)\Psi^{(1)}(z\hbar^{-1})\Lambda^{\bullet}(N_1)^{-1} 
\end{align}
were $J(z)$ is given by (\ref{fusi}). 
\end{prop}

\begin{proof} Consider the following modified capping operators \bean \label{tildnotation} \Tilde{\Psi}^{(2)}(z)= \Delta_{(1,1)}^{-1} \Lambda^{\bullet}(N_2) \Psi^{(2)}(z) \Lambda^{\bullet}(N_2)^{-1} \Delta_{(1,1)}, \ \ \ \Tilde{\Psi}^{(1)}(z)=   \Psi^{(1)}(z) 
\eean 
obtained by conjugating our capping operator by the normal weights. In \cite{Sm16} it was shown that for these normalizations we have:
\bean \label{smequa}
\lim\limits_{a\rightarrow 0} \Tilde{\Psi}^{(2)}(z) = Y(z) \Tilde{\Psi}^{(1)}(z\hbar) \otimes \Tilde{\Psi}^{(1)}(z\hbar^{-1})
\eean
where $Y(z)$ is an operator acting on K-theory as the following element of  $\mathcal{U}_{\hbar}({\ddot{\mathfrak{gl}_1}})$:
$$Y(z)=\exp\left(-\sum_{k=1}^{\infty}\frac{n_k K^{-k}\otimes K^{k}}{1-z^{-k}K^{-k}\otimes K^k} \alpha^{0}_{-k}\otimes \alpha^{0}_{k}\right)$$
Since all the terms in the exponent of (\ref{fusi}) commute with each other, it is elementary to check that
\bean \label{JJop}
J(z)J(0)^{-1}=\exp{\left(\sum_{k=1}^{\infty}\frac{n_kK^{-k}\otimes K^{k}}{1-z^{-k}K^{-k}\otimes K^k}\alpha^{0}_{-k}\otimes \alpha^{0}_{k}\right)}=J(z^{-1})^{-1} = Y(z)^{-1}
\eean 
Thus (\ref{smequa}) becomes:
$$
J(z)J(0)^{-1} \times \lim\limits_{a\to 0} \,\Tilde{\Psi}^{(2)}(z)=\Tilde{\Psi}^{(1)}(z\hbar)\otimes \Tilde{\Psi}^{(1)}(z\hbar^{-1})$$
which is what we want to prove. 
\end{proof}

\subsection{Limits of Vertex Functions}

Let $\mathcal{V}$ denote the tautological bundle of rank $n$ on $\mathcal{M}(n,2)$. Let $c_k(\mathcal{V})$ be the diagonal matrix of multiplication by $c_k(\mathcal{V})$ in the basis of fixed points of $\mathcal{M}(n,2)$. Abusing notations as before,  we denote by $c_k(\mathcal{V}) \otimes 1$  the operator acting in the K-theory of $\mathcal{M}(n,2)^{\mathbb{C}^{*}} = \coprod\limits_{n_1+n_2=n} \textrm{Hilb}^{n_1}(\mathbb{C}^2) \times \textrm{Hilb}^{n_2}(\mathbb{C}^2)$ are the operator of multiplication by $c_k(\mathcal{V_1})$ where  denote the tautological bundle for the first component $\textrm{Hilb}^{n_1}(\mathbb{C}^2)$.  Again, we assume that $c_k(\mathcal{V}) \otimes 1$ is a diagonal matrix considering this operator in the basis of torus fixed points. With these notations we have:

\begin{prop} \label{prop4}
We have the following limit
\bean \label{limc}
\lim\limits_{a\to 0} \, c_k(\mathcal{V})= c_k(\mathcal{V}) \otimes 1
\eean 
\end{prop}

\begin{proof} 
The eigenvalue of $c_k(\mathcal{V})$ at a torus fixed point $(\lambda_1 \lambda_2)\in \mathcal{M}(n,2)^{T}$ equals:
$$\left.c_k(\mathcal{V})\right|_{(\lambda_1 \lambda_2)}=e_k(x_{\square}= \varphi_{\lambda_1,\lambda_2}(\square))$$
where $e_k$ is the $k$-the elementary symmetric polynomial evaluated at (\ref{weights}). By our rank decomposition $r=r_1+ar_2$ we can see that the only factor that carries the $a$ term are the boxes in the second rank 1 space. This corresponds to a limit of $\varphi_{\lambda_1,\lambda_2}(\square)= a_{n(\square)}t_1^{y(\square)}t_2^{x(\square)}$ as follows:

    $$\lim\limits_{a\rightarrow 0}\varphi_{\lambda_1,\lambda_2}(\square) =\begin{cases} 
      \varphi_{\lambda_1}(\square)  & \square\in \lambda_1 \\
      0 & \square\in\lambda_2
   \end{cases}$$
Thus $\lim\limits_{a\to 0} \, \left.c_k(\mathcal{V})\right|_{(\lambda_1 \lambda_2)}= \lim\limits_{a\to 0} \, \left.c_k(\mathcal{V})\right|_{\lambda_1}$, which gives the statement of the proposition. 
\end{proof}

Let $\textup{V}^{\tau,(r)}(z)$ denote the rank $r$ bare vertex function with descendent $\tau$.
\begin{prop}\label{vertexlimit}
We have the following limit of vertex functions:
$$
\lim\limits_{a\to 0} \, \CO(-1)\textup{V}^{c_k(\mathcal{V}),(2)}(z)= \CO(-1)\textup{V}^{c_k(\mathcal{V}),(1)}(-z t_1 t_2 q) \otimes \CO(-1)\textup{V}^{1,(1)}(- z/(t_1t_2 q))
$$
\end{prop}

\begin{proof}
This is Theorem 4 in \cite{Sm16}.
\end{proof}

\subsection{A relation between capped descendents}

Let us consider the following descendent:
\begin{equation}\label{descendent}
\tau_n(u)= (u+x_1) \dots (u+x_n) = \sum\limits_{i=0}^{n} u^{n-i} c_{i}(\mathcal{V})
\end{equation}
where $c_{i}(\mathcal{V})$ denotes the $i$-the Chern class of the tautological bundle $\mathcal{V}$, and $x_1,...,x_n$ the Chern roots of $\mathcal{V}$ over $\mathcal{M}(n,r)$. Explicitly, it is the $i$-th symmetric polynomial in the Chern roots. 
We will also need the dual generating function
\begin{equation} \label{dualdes}
\bar{\tau}_n(u)= (1+u x_1^{-1}) \dots (1+u x_n^{-1}) = \sum\limits_{i=0}^{n} u^{i} c_{i}(\mathcal{V}^{*})
\end{equation}
where $c_{i}(\mathcal{V}^{*})$ is the $i$-th Chern class of the dual bundle $\mathcal{V}^{*}$.  We note that $\CO(1)=x_1\dots x_n$ and therefore:
$$
\tau_n(u) = \CO(1) \bar{\tau}_n(u)
$$
We recall that the operator $\textsf{M}(z)$ is the operator of quantum multiplication by $\CO(1)$ and the above identity can be upgraded to capped descendent vertices:
\bean \label{cappedrelat}
\Hat{\textup{V}}^{\tau_n(u)}(z) = \textsf{M}(z q) \Hat{\textup{V}}^{\bar{\tau}_n(u)}(z)
\eean
Our goal is to compute the generating function
\bean \label{compons}
F(z,t_1,t_2, q,u,y)= \sum\limits_{n=0}^{\infty}\, \Hat{\textup{V}}^{\bar{\tau}_n(u)}(z) \, y^n
\eean

\subsection{Limiting linear equation}
Let $\Hat{\textup{V}}^{\tau(u)}(z)$ denote the capped vertex function with descendents considered as an element of the Fock space, i.e., its degree $n$ component is given by  $\Hat{\textup{V}}^{\tau_n(u)}(z)$. 

\begin{prop}
In the tensor square $\textsf{\textup{Fock}}^{\otimes 2}$ we have the following identity 
\bean \label{limiteq}
\widehat{\textup{V}}^{(\tau(u))}(-z\hbar q) \otimes 1 = (J(z)J(0)^{-1}) \Big(\bar{\tau}(u) \otimes 1\Big)
\eean
where $\bar{\tau}(u) $ is the dual descendents (\ref{dualdes})
\end{prop}
\begin{proof}
From Theorem \ref{LargeFrame} in for rank $r=2$ we have 
\bean \label{rank2cap}
\CO(-1)  c_k(\mathcal{V}^*)=\Lambda^{\bullet}(T\mathcal{M}(n,2)^{\vee})\Psi(z/q^2) \Lambda^{\bullet}(T\mathcal{M}(n,2)^{\vee})^{-1} \textup{V}_O^{(c_k(\mathcal{V}^*))}(z)
\eean 
where $\textup{V}_O^{(c_k(\mathcal{V}^*))}(z)$ denotes the non-normalized vertex function whose expansion in the fixed point basis has the form:
$$\textup{V}_O^{(c_k(\mathcal{V}^*))}(z)=\CO(1) c_k(\mathcal{V}^*) +O(z).$$
The tangent weights decompose as follows $\Lambda^{\bullet}(T\mathcal{M}(n,2)^{\vee})=\Lambda^{\bullet}(TH^{\vee})\oplus  \Lambda^{\bullet}(N_2)$, where $\Lambda^{\bullet}(TH^{\vee})$ are the tangent weights of the Hilbert scheme $\mathcal{M}(n,2)^{A}=\coprod\limits_{n_1+n_2=n} \textup{Hilb}^{n_1}(\mathbb{C}^2) \times  \textup{Hilb}^{n_2}(\mathbb{C}^2) $ and $\Lambda^{\bullet}(N_2)$ are the weights of its normal bundle  inside $\mathcal{M}(n,2)$. Thus:
$$\CO(-1) c_k(\mathcal{V}^*)=\Lambda^{\bullet}(TH^{\vee})  \Lambda^{\bullet}(N_2)  \Psi^{(2)}(z/q^2) \Lambda^{\bullet}(N_2)^{-1}\Lambda^{\bullet}(TH^{\vee})^{-1}\CO(-1)\textup{V}^{(c_k(\mathcal{V}^*))}(z)$$
Or, equivalently 
\begin{align*}
  \CO(-1)   c_k(\mathcal{V}^*)&=\Lambda^{\bullet}(TH^{\vee})  \Lambda^{\bullet}(N_2)\Delta_{(1,1)}^{-1}\Delta_{(1,1)}  \Psi^{(2)}(z/q^2)\\
    &\times\Lambda^{\bullet}(N_2)^{-1}\Lambda^{\bullet}(TH^{\vee})^{-1}\Delta_{(1,1)}\Delta_{(1,1)}^{-1} \CO(-1)\textup{V}^{(c_k(\mathcal{V}^*))}(z)
\end{align*}
Multiplying by $\CO(1)$, and noting that this operator commutes with $\Lambda^{\bullet}(TH^{\vee})$ and $\Delta_{(1,1)}$  we obtain:
\begin{align*}
   c_k(\mathcal{V}^*)&=\Lambda^{\bullet}(TH^{\vee})  \Delta_{(1,1)}\CO(1) a^{-2n}  \Lambda^{\bullet}(N_2)\Delta_{(1,1)}^{-1}\Psi^{(2)}(z)\Lambda^{\bullet}(N_2)^{-1}\Delta_{(1,1)}\\
    &\times\Lambda^{\bullet}(TH^{\vee})^{-1}\Delta_{(1,1)}^{-1} \CO(-1)a^{2n}\textup{V}^{(c_k(\mathcal{V}^*))}(z)
\end{align*}
where we divided and multiplied factors by $a^{2n}$. 
In our notation (\ref{tildnotation}) this simplifies to
\begin{align*}
 c_k(\mathcal{V}^*)&=\Lambda^{\bullet}(TH^{\vee})  \Delta_{(1,1)}\CO(1) a^{-2n}  \Tilde{\Psi^{(2)}}(z/q^2) \Lambda^{\bullet}(TH^{\vee})^{-1}\Delta_{(1,1)}^{-1} \CO(-1)a^{2n}\textup{V}^{(c_k(\mathcal{V}^*))}(z)
\end{align*}
To compute the limit $a\rightarrow 0$, note that the Hilbert Scheme weights $\Lambda^{\bullet}(TH^{\vee})$ do not depend on $a$. Next, by Proposition \ref{prop1}
$$\lim\limits_{a\to 0} \,\Delta_{(1,1)}\CO(1) a^{-2n} =\hbar^{-n} \left\{{\hbar}^{-n_2}\CO(1)\otimes 1\right\}$$
and by Proposition \ref{limitcapping} the limit of the modified capping operator was computed as 

$$ \lim\limits_{a\to 0} \,\Tilde{\Psi}^{(2)}(z/q^2)=\left(J(z)J(0)^{-1} \right)^{-1}\Tilde{\Psi}^{(1)}(-z\hbar/q^2)\otimes \Tilde{\Psi}^{(1)}(-z\hbar^{-1}/q^2)$$
Finally by Proposition \ref{vertexlimit}

$$
\lim\limits_{a\to 0} \, \textup{V}^{c_k(\mathcal{V}),(2)}(z)= \textup{V}^{c_k(\mathcal{V}),(1)}(- z \hbar q) \otimes \textup{V}^{1,(1)}(- z\hbar^{-1} /q)
$$
where we have factored out the $\CO(-1)$ and thus our vertex function begins with $1$ and has a limit. Combining all these terms together, and using (\ref{limc}) in the limit $a\to 0$ we arrive at 
\begin{align*}
   c_k(\mathcal{V}^*)\otimes 1= &\Lambda^{\bullet}(TH^{\vee}) \times (J(0) J(-z)^{-1})^{-1} \times \tilde{\Psi}^{(1)}(-z\hbar q^{-2})\otimes \tilde{\Psi}^{(1)}(-z\hbar^{-1} q^{-2}) \\
    &\times \Lambda^{\bullet}(TH^{\vee})^{-1}\times \textup{V}^{(c_k(\mathcal{V}^*)),(1)}(-z \hbar q ) \otimes \textup{V}^{1,(1)}(-z\hbar^{-1}/q )
\end{align*}
In the first tensor component of the last expression we have
\begin{align*}
\tilde{\Psi}^{(1)}(-z\hbar q^{-2})\frac{1}{\Lambda^{\bullet}(TH^{\vee})}\textup{V}^{(c_k(\mathcal{V}^*)),(1)}(-z \hbar q )=  \mathsf{M}^{(1)}(-z \hbar) \tilde{\Psi}^{(1)}(-z\hbar)\frac{1}{\Lambda^{\bullet}(TH^{\vee})}\textup{V}^{(c_k(\mathcal{V}^*)),(1)}(-z \hbar q )
\end{align*}
where we used QDE (\ref{qdeq}). Using (\ref{cappdef}) for $r=1$ we also write it as: 
\begin{align*}
\mathsf{M}^{(1)}(-z \hbar) \tilde{\Psi}^{(1)}(-z\hbar)\frac{1}{\Lambda^{\bullet}(TH^{\vee})}\textup{V}^{(c_k(\mathcal{V}^*)),(1)}(-z \hbar q ) =\mathsf{M}^{(1)}(-z \hbar) \widehat{\textup{V}}^{(c_k(\mathcal{V}^*)),(1)}(-z \hbar q ) & \\
=\widehat{\textup{V}}^{(c_k(\mathcal{V})),(1)}(-z \hbar q )
\end{align*}
where the last equality is  (\ref{cappedrelat}). 
$$ \tilde{\Psi}^{(1)}(-z\hbar q^{-2})\frac{1}{\Lambda^{\bullet}(TH^{\vee})}\textup{V}^{(c_k(\mathcal{V})),(1)}(-z \hbar q )= \frac{1}{\Lambda^{\bullet}(TH^{\vee})}\, \widehat{\textup{V}}^{(c_k(\mathcal{V})),(1)}(-z \hbar q)$$
Similarly, in the second tensor component by (\ref{cappdef}) for $r=1$ we have
$$ \tilde{\Psi}^{(1)}(-z\hbar^{-1}
q^{-2})\frac{1}{\Lambda^{\bullet}(TH^{\vee})}\textup{V}^{1,(1)}(-z \hbar^{-1}/q ) = \widehat{\textup{V}}^{1,(1)}(-z \hbar^{-1}/q )=1$$
where we used that the capped vertex for descendents $\tau=1$ is trivial by large framing vanishing Theorem any rank.   Combining all this together, in the limit $a\to 0$ we obtain the following relation:
$$(J(z)J(0)^{-1})^{-1}\left(\widehat{\textup{V}}^{(c_k(\mathcal{V}))}(-z\hbar q) \otimes 1\right)=c_k(\mathcal{V}^{*}) \otimes 1,$$
or
$$
\widehat{\textup{V}}^{(c_k(\mathcal{V}))}(-z\hbar q) \otimes 1 = (J(z)J(0)^{-1}) \Big(c_k(\mathcal{V}^{*}) \otimes 1\Big)
$$
Summing over $k$ gives the desired identity. 
\end{proof}

\section{Derivation of Descendents formula}
We find a formula for capped descendent vertex $\widehat{\textup{V}}^{(\tau(u))}(z)$ by solving the identity (\ref{limiteq}) for the first tenor component. This can be done as follows.

Let $H_{\lambda}$ denote the Macdonald polynomial in Haiman's normalization. Our parameters are related to $q,t$ via the relation $q=1/t_1^2, t=1/t_2^2$. 

\begin{prop}\label{macidentity} The following identity holds in the Fock space (\ref{fockdef}):
      $$ \sum_{\lambda}\frac{H_{\lambda}}{\Lambda^{\bullet}(T_{\lambda}\text{Hilb}^{|\lambda|}(\mathbb{C}^2))}y^{|\lambda|}=\exp\left(\sum\limits_{k=1}^{\infty} \dfrac{ y^k \hbar^{2 k} p_{k}}{k (1-t_1^{2k})(1-t_2^{2k})}  \right) $$
\end{prop}
\begin{proof}
   This is the kernel identity for Macdonald scalar product. See \cite{Me16}, with $Y=1$.
\end{proof}
\begin{prop}\label{mellitT}
We have the following generating function given by \ref{descendent}:
\bean  \label{genvert0}
\sum\limits_{n=0}^{\infty}  \bar{\tau}_n(u)  y^n  = \exp\left( \sum\limits_{k=1}^{\infty} \dfrac{y^k p_{k}}{k (1-t_1^{2k})(1-t_2^{2k})} \, \hbar^{2 k}  \Big( 1 -  u^k \Big) \right)
\eean
\end{prop}
\begin{proof}
 See corollary 6.4 in \cite{Me16}   
\end{proof}
\noindent Rescaling  $p_i\to p_i/a_i$ and sending $a\to \infty$, from this proposition we obtain the following twisted version of the above normalized Macdonald generating series:
\begin{equation}\label{OSum}
\sum\limits_{n=0}^{\infty} \mathcal{O}_{\textup{Hilb}^n(\mathbb{C}^2)} \,  y^n  = \exp\left( \sum\limits_{k=1}^{\infty} \dfrac{(-1)^k y^k \hbar^{2 k} p_{k}}{k (1-t_1^{2k})(1-t_2^{2k})}   \right)
\end{equation}
Thus, in the tensor square of the Fock space $\textsf{\textup{Fock}}^{\otimes 2}$ we have
$$
    \label{twogen}
\sum\limits_{n_1,n_2=0}^{\infty} 
\, (\tau_{n_1}(u) \otimes \mathcal{O}_{\textup{Hilb}^{n_2}(\mathbb{C}^2)}) y^{n_1+n_2} =
\exp\left( \sum\limits_{k=1}^{\infty} \dfrac{y^k \hbar^{2 k} ( ( 1 -  u^k ) p^{(1)}_{k}+(-1)^k p^{(2)}_k) }{k (1-t_1^{2k})(1-t_2^{2k})}    \right)
$$ 
where the superscripts in $p^(i)_{k}$, $i=1,2$ denote the components of the first and the second factors in the tensor product $\textsf{\textup{Fock}}^{\otimes 2}$. 

From (\ref{JJop}) and $0$-slope Heisenberg algebra action on the Fock space (\ref{alphazer}), we find that $J(z)J(0)^{-1}$ acts  in $\textsf{\textup{Fock}}^{\otimes 2}$ as follows:
$$
J(z)J(0)^{-1}: \, p^{(1)}_k \to p^{(1)}_k, \ \ \  p^{(2)}_k \to p^{(2)}_k  + (-1)^k z^k \hbar^{k} \dfrac{\hbar^k - \hbar^{-k}}{1-z^k} p^{(1)}_k 
$$
Thus, applying this operator to (\ref{twogen}) we obtain
$$
J(z)J(0)^{-1} \sum\limits_{n_1,n_2=0}^{\infty} 
\, (\tau_{n_1}(u) \otimes \mathcal{O}_{\textup{Hilb}^{n_2}(\mathbb{C}^2)}) y^{n_1+n_2}=
$$
$$
 = 
\exp\left( \sum\limits_{k=1}^{\infty} \dfrac{y^k  }{k (1-t_1^{2k})(1-t_2^{2k})} ( ( 1 -  u^k ) \hbar^{2 k} p^{(1)}_{k}+(-1)^k\hbar^{2 k} p^{(2)}_k+z^k \hbar^{ k} \dfrac{\hbar^k - \hbar^{-k}}{1-z^k} p^{(1)}_k  )     \right)
$$
Taking the first component in the tensor product corresponds to $p^{(2)}_k=0$, which gives:
$$
\exp\left( \sum\limits_{k=1}^{\infty} \dfrac{y^k  }{k (1-t_1^{2k})(1-t_2^{2k})} \Big( ( 1 -  u^k ) \hbar^{2 k} p^{(1)}_{k}+ \hbar^{k} z^k \dfrac{\hbar^k - \hbar^{-k}}{1-z^k} p^{(1)}_k  \Big)     \right)
$$
For the generating function (\ref{compons}) we thus proved the following theorem:
\begin{thm}\label{mainformula}
\bean \label{finres}
F(-z  \hbar q,t_1,t_2, q,u,y) =  \nonumber \exp\left( \sum\limits_{k=1}^{\infty} \dfrac{y^k  }{k (1-t_1^{2k})(1-t_2^{2k})} ( ( 1 -  u^k ) \hbar^{2 k} p^{(1)}_{k}+z^k \hbar^{ k} \dfrac{\hbar^k - \hbar^{-k}}{1-z^k} p^{(1)}_k  )\right)
\eean
\end{thm}

\begin{cor}
The capped vertices $\Hat{\textup{V}}^{(c_k)}(z q)$ do not depend on $q$. 
\end{cor}

The coefficients of $y$ Taylor expansion of this generating function are explicitly rational functions of the quantum parameter $z$. Thus, we also obtain:
\begin{cor}\label{rationalcor} The capped vertices 
$\widehat{\textup{V}}_{\textup{Hilb}^n(\mathbb{C}^2)}^{c_k}(z)$ with descendents given by Chern classes $c_k$ are Taylor series expansions of rational functions in the quantum parameter $z$. 
\end{cor}

On behalf of all authors, the corresponding author states that there is no conflict of interest.
\nocite{*}
\bibliographystyle{plain}
\bibliography{cappedvertexhilb.bib}

\noindent
Jeffrey Ayers\\
Department of Mathematics, University
of North Carolina at Chapel Hill\\ Chapel Hill, NC 27599-3250, USA\\
jeff97@live.unc.edu\\

\noindent
Andrey Smirnov\\
Department of Mathematics, University
of North Carolina at Chapel Hill\\ Chapel Hill, NC 27599-3250, USA\\
asmirnov@email.unc.edu
\end{document}